\documentclass[12pt,reqno]{article}

\usepackage[usenames]{color}
\usepackage{amssymb}
\usepackage{amsmath}
\usepackage{amsthm}
\usepackage{amsfonts}
\usepackage{amscd}
\usepackage{graphicx}
\usepackage{xcolor}

\usepackage[colorlinks=true,
linkcolor=webgreen,
filecolor=webbrown,
citecolor=webgreen]{hyperref}

\definecolor{webgreen}{rgb}{0,.5,0}
\definecolor{webbrown}{rgb}{.6,0,0}

\usepackage{color}
\usepackage{fullpage}
\usepackage{float}

\usepackage{graphics}
\usepackage{latexsym}
\usepackage{epsf}
\usepackage{breakurl}
\usepackage{fullpage}

\newcommand{\seqnum}[1]{\href{https://oeis.org/#1}{\rm \underline{#1}}}

\DeclareMathOperator{\factoreq}{factoreq}
\DeclareMathOperator{\spec}{spec}
\DeclareMathOperator{\speclen}{speclen}
\DeclareMathOperator{\maxspec}{maxspec}
\def\andd{\, \wedge \, }

\begin{document}

\theoremstyle{plain}
\newtheorem{theorem}{Theorem}
\newtheorem{corollary}[theorem]{Corollary}
\newtheorem{lemma}[theorem]{Lemma}
\newtheorem{proposition}[theorem]{Proposition}

\theoremstyle{definition}
\newtheorem{definition}[theorem]{Definition}
\newtheorem{example}[theorem]{Example}
\newtheorem{conjecture}[theorem]{Conjecture}

\theoremstyle{remark}
\newtheorem{remark}[theorem]{Remark}

\title{Note on a Fibonacci Parity Sequence}

\author{Jeffrey Shallit\\
School of Computer Science\\
University of Waterloo\\
Waterloo, ON  N2L 3G1\\
Canada\\
\href{mailto:shallit@uwaterloo.ca}{\tt shallit@uwaterloo.ca}}

\maketitle

\begin{abstract}
Let ${\bf ftm} = {\tt 0111010010001}\cdots$ be the analogue of
the Thue-Morse sequence in Fibonacci representation.  In this
note we show how,
using the {\tt Walnut} theorem-prover,
to obtain a measure of its complexity,
previously studied by Jamet, Popoli, and Stoll.
We strengthen one of their theorems and disprove one of their
conjectures.
\end{abstract}

\section{Introduction}

Recall that the Fibonacci numbers $(F_n)$ are defined by
$$ F_0 = 0, \quad F_1 = 1, \quad \text{ and } \quad  F_n = F_{n-1} + F_{n-2}  \quad \text{for $n \geq 2$} .$$
In Fibonacci representation (aka Zeckendorf representation),
we express a natural number
$n$ uniquely as a sum of non-adjacent Fibonacci numbers:
$n = \sum_{2 \leq i \leq t} e_i F_i$, where
$e_i \in \{0,1\}$  and $e_i e_{i+1} = 0$ for $2 \leq i < t$.
See, for example, \cite{Lekkerkerker:1952,Zeckendorf:1972}.

The sum of the Fibonacci bits is $s_F (n) = \sum_{2 \leq i \leq t} e_i$,
and the so-called Fibonacci-Thue-Morse sequence $\bf ftm$ is
then defined as ${\bf ftm}[n] = s_F(n) \bmod 2$.   Here are the
first few terms of this binary sequence, which is sequence
\seqnum{A095076} in the {\it On-Line Encyclopedia of Integer
Sequences} (OEIS) \cite{Sloane}:
\begin{center}
\begin{tabular}{c|cccccccccccccccc}
$n$ & 0 & 1 & 2 & 3 & 4 & 5 & 6 & 7 & 8 & 9 & 10 & 11 & 12 & 13 \\
\hline
${\bf ftm}[n]$ & {\tt 0} & {\tt 1} & {\tt 1} & {\tt 1} & {\tt 0} & {\tt 1} & {\tt 0} & {\tt 0} & {\tt 1} & {\tt 0} & {\tt 0} & {\tt 0} & {\tt 1} & {\tt 1}  
\end{tabular}
\end{center}
This sequence was previously studied by Ferrand \cite{Ferrand:2007}
and the author \cite{Shallit:2021}, to name just two appearances.

Recently $\bf ftm$ appeared in a paper of Jamet, Popoli, and Stoll
\cite{Jamet&Popoli&Stoll:2021}, as follows.   A factor (contiguous
block) $x$ in a 
(finite or infinite) binary sequence $s$ is said to be {\it special\/}
if both $x{\tt 0}$ and $x{\tt 1}$ appear in $s$.   
Now let $\bf s$ be an infinite sequence, and define
$f_{\bf s} (n)$ be the length of the longest special factor
in a length-$n$ prefix of $\bf s$.

For example, here are the first few terms of $f_{\bf s}(n)$
when ${\bf s} = {\bf ftm}$:
\begin{center}
\begin{tabular}{c|ccccccccccccccccccccc}
$n$ &  2 & 3 & 4 & 5 & 6 & 7 & 8 & 9 & 10 & 11 & 12 & 13 &14 & 15 & 16 & 17 & 18 & 19\\
\hline
$f_{\bf ftm}(n)$ & 0& 0& 0& 2& 2& 2& 2& 2& 2& 2& 4& 4& 4& 4& 4& 4& 4& 5
\end{tabular}
\end{center}

Jamet et al.\ studied a certain
measure of the complexity of a binary sequence $\bf s$, called
maximum order complexity $M_{\bf s}(n)$, and observed that by a result of 
Jansen \cite{Jansen:1989}, we have the relationship
$M_{\bf s}(n) = f_{\bf s}(n)+1$.

Letting $\alpha = (1+\sqrt{5})/2$ be the golden ratio, their
Theorem 1.1 is the following:
\begin{theorem}
There exists $N_0$ such that for all $n > N_0$ we have
$M_{\bf ftm} (n) \geq {n \over {\alpha+\alpha^3}} + 1$.
\label{thm1}
\end{theorem}
They also gave the following conjecture, Conjecture 3.1 in their paper:
\begin{conjecture}
$M_{\bf ftm} (n) \sim {n \over {1 + \alpha^2}}$.
\label{conj2}
\end{conjecture}
In this note we first obtain an exact formula for
$f_{\bf ftm} (n)$ (our Theorem~\ref{main}), using the
{\tt Walnut} theorem-prover.  
(For more about {\tt Walnut}, see \cite{Mousavi:2016} and \cite{Shallit:2022}.)
Then, using this result, we disprove Conjecture~\ref{conj2}.
Finally,
we show how to reprove Theorem~\ref{thm1} and also determine
the explicit value of $N_0$ in that theorem.

\section{First-order formulas}
\label{two}

The idea behind our approach is that we can express the assertion that
$m = f_{\bf s}(n)$ as a formula in first-order logic.
This is done as follows:
\begin{itemize}
\item $\factoreq(i,j,n)$ asserts that ${\bf s}[i..i+n-1] = {\bf s}[j..j+n-1]$;
\item $\spec(i,m,n)$ asserts that there is a $j$ such that
$\factoreq(i,j,m)$ but ${\bf s}[i+m] \not= {\bf s}[j+m]$, and 
$\max(i+m,j+m) <n$;
\item $\speclen(m,n)$ asserts that there is a special factor of
length $m$ in ${\bf s}[0..n-1]$, the length-$n$ prefix of $\bf s$;
\item $\maxspec(m,n)$ asserts that the length of the longest special factor occurring
in a prefix of length $n$ is $m$.
\end{itemize}

We can translate these into first-order formulas as follows:
\begin{align*}
\factoreq(i,j,n) &= \forall u, v \ (i+v=j+u \andd u\geq i \andd u<i+n) \implies
{\bf s}[u] = {\bf s}[v] \\
\spec(i,m,n) &= \exists j\ \factoreq(i,j,m) \andd {\bf s}[i+m] \not= {\bf s}[j+m]
\andd i+m<n \andd j+m < n \\
\speclen(m,n) &= \exists i \ \spec(i,m,n) \\
\maxspec(m,n) &= \speclen(m,n) \andd \forall j\ (j>m) \implies \neg \speclen(j,n).
\end{align*}

It now follows from known results
\cite{Charlier&Rampersad&Shallit:2012}
that if $\bf s$ is a (generalized)
automatic sequence, then there is an effective algorithm for computing
an automaton that accepts the representation of those pairs of
natural numbers $m,n$ for which $m$ is the length of the longest
special factor in ${\bf s}[0..n-1]$.  In particular, the sequence
${\bf ftm}$ is Fibonacci-automatic, so we can find this automaton
for the pairs $(m,n)$ in Fibonacci representation \cite{Mousavi&Schaeffer&Shallit:2016}.  This is done
in the next section.

\section{Translation to {\tt Walnut}}

We now translate the first-order formulas to {\tt Walnut}.  The first
step is to obtain an automaton for the sequence $\bf ftm$ itself.
We can either get this directly from the morphism and coding
in the paper of Jamet et al., or compute it with {\tt Walnut}.  For
this latter approach we use the {\tt Walnut} commands
\begin{verbatim}
reg odd1 msd_fib "0*(10*10*)*10*":
def zecksum "?msd_fib (n>=0) & $odd1(n)":
combine FTM zecksum:
\end{verbatim}
The first line defines a regular expression for a binary string having
an odd number of $\tt 1$'s; the second asserts that an integer $n$ has
an odd number of $\tt 1$'s in its Fibonacci representation; and the third line
turns this automaton into an automaton with output (DFAO).

Next, we translate the first-order formulas from Section~\ref{two}
into {\tt Walnut}:
\begin{verbatim}
def factoreq "?msd_fib Au,v (i+v=j+u & u>=i & u<i+n) => FTM[u]=FTM[v]":
def spec "?msd_fib (i+m<n) & Ej (j+m<n) & $factoreq(i,j,m) &
   FTM[i+m] != FTM[j+m]":
def speclen "?msd_fib Ei $spec(i,m,n)":
def maxspec "?msd_fib $speclen(m,n) & Aj (j>m) => ~$speclen(j,n)":
\end{verbatim}
By comparing these with the formulas above, we see that they are
more-or-less a direct translation.

When we run all these commands through {\tt Walnut}, the last one
produces an automaton of 17 states depicted in Figure~\ref{ms}.
\begin{figure}[htb]
\begin{center}
\includegraphics[width=6.7in]{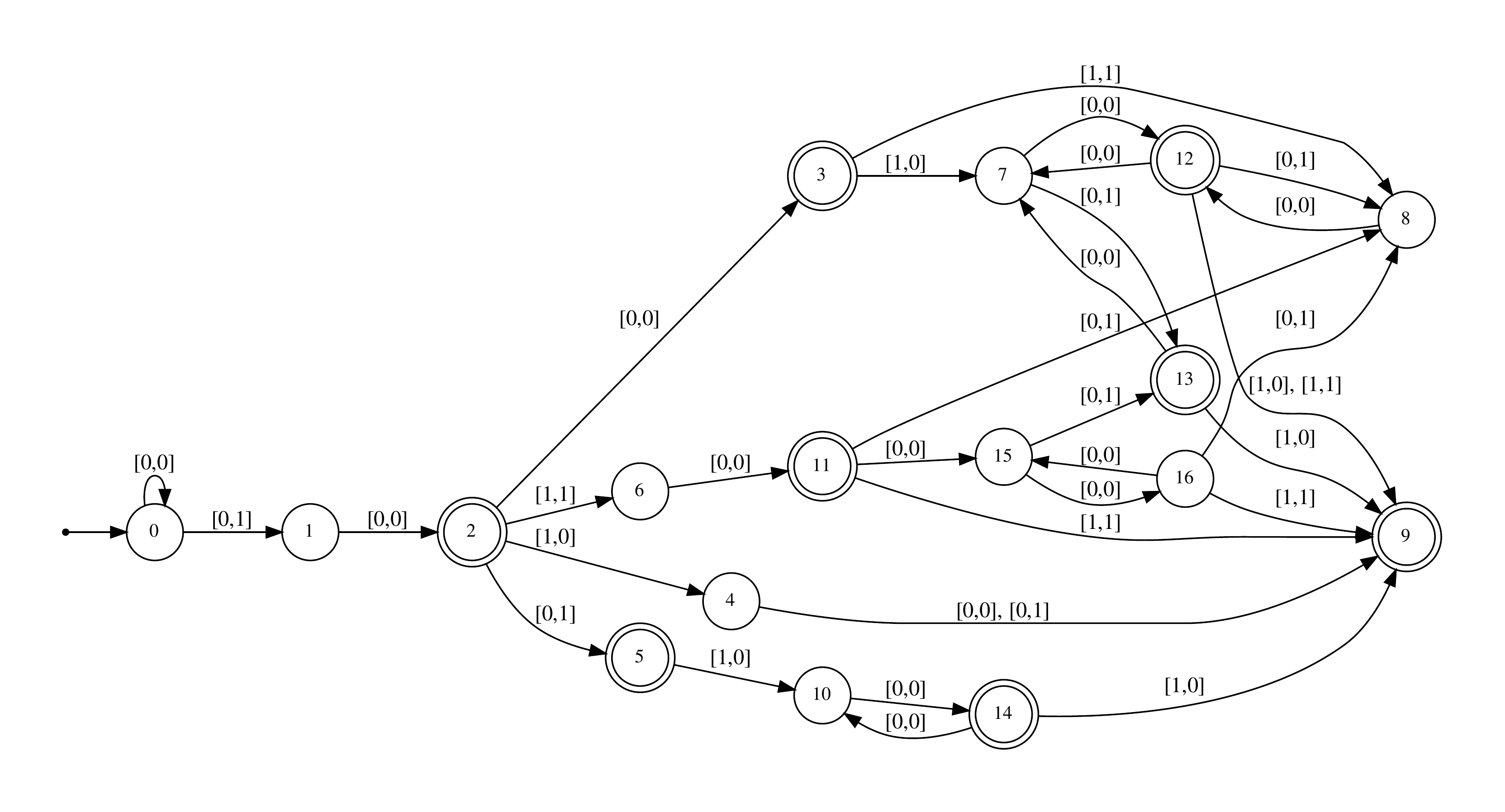}
\end{center}
\caption{Automaton accepting Fibonacci representation of
$(\maxspec(n), n)$.}
\label{ms}
\end{figure}
This shows that the map $n \rightarrow f_{\bf ftm} (n)$ is
``Fibonacci-synchronized''; see, for example, \cite{Shallit:2021b}.

We can use this automaton to prove the following ``exact'' formula
for $f_{\bf ftm}(n)$.   Define the Lucas numbers $(L_n)$ by
$L_0 = 2$, $L_1 = 1$, and $L_n = L_{n-1} + L_{n-2}$ for $n \geq 2$.
Then we have the following result:
\begin{theorem}
Suppose $i \geq 4$ and $n \geq 8$. If
$L_i < n \leq L_{i+1}$, then
$$f_{\bf ftm} (n) = \begin{cases}
	F_{i-1}, & \text{for $i$ even;} \\
	F_{i-1}+1, & \text{for $i$ odd.}
	\end{cases}
$$
\label{main}
\end{theorem}

\begin{proof}
We can use {\tt Walnut} to verify the statement.  In order to 
determine which case applies,
we need to be able to compute $F_{i-1}, L_i,$ and $L_{i+1}$ simultaneously.
In Fibonacci representation, this can be done with the following
four simple observations:
\begin{itemize}
\item For $i \geq 2$, the number
$F_i$ has Fibonacci representation $1 \overbrace{{\tt 0}\cdots {\tt 0}}^{i-2}$.
\item If $F_i$ and $F_j$ are two Fibonacci numbers
with $F_i < F_j \leq 2F_i$, then $j = i+1$.
\item For $i \geq 1$ we have $L_i = 2F_{i-1} + F_i$.
\item For $i \geq 1$ we have $L_{i+1} = F_{i-1} + 3F_i$.
\end{itemize}

We can assert that $(x,y) = (F_{2i-1}, F_{2i})$ for some $i \geq 2$
by asserting that $(x)_F$ ends in an odd number of $\tt 0$'s,
$(y)_F$ ends in an even number of $\tt 0$'s, and $x<y$ and $2x \geq y$.
Similarly, we can assert that $(x,y) = (F_{2i}, F_{2i+1})$ for some $i \geq 1$
by asserting that $(x)_F$ ends in an even number of $\tt 0$'s,
$(y)_F$ ends in an odd number of $\tt 0$'s, and $x<y$ and $2x \geq y$.
This can be done with the following {\tt Walnut} commands:
\begin{verbatim}
reg isevenfib msd_fib "0*1(00)*":
reg isoddfib msd_fib "0*10(00)*":
def fiboddeven "?msd_fib $isoddfib(x) & $isevenfib(y) & x<y & (2*x)>=y":
def fibevenodd "?msd_fib $isevenfib(x) & $isoddfib(y) & x<y & (2*x)>=y":
\end{verbatim}

We can now check the claim of the theorem as follows:
\begin{verbatim}
def check_i_even "?msd_fib An,x,y ($fiboddeven(x,y) & n>=8 & 2*x+y<n & 
   n<=x+3*y) => $maxspec(x,n)":
def check_i_odd "?msd_fib An,x,y ($fibevenodd(x,y) & n>=8 & 2*x+y<n & 
   n<=x+3*y) => $maxspec(x+1,n)":
\end{verbatim}
and both return {\tt TRUE}.
Here $x$ plays the role of $F_{i-1}$ and $y$ plays the role of $F_i$.
\end{proof}

Because $f_{\bf ftm} (n)$ is constant on longer and longer intervals of $n$,
it becomes easy to compute $\liminf_{n \rightarrow \infty} f_{\bf ftm} (n)/n$
and $\limsup_{n \rightarrow \infty} f_{\bf ftm} (n)/n$.

In particular, if $L_i < n \leq L_{i+1}$ for $i$ even, $i \geq 4$,
Theorem~\ref{main} implies that
the quotient $f_{\bf ftm}(n)/n$ is minimized at
$n = L_{i+1}$, where it takes the value $F_{i-1}/L_{i+1}$,
and maximized at $n = L_i + 1$, where it takes the value
$F_{i-1}/(L_i + 1)$.

Similarly, if $L_i < n \leq L_{i+1}$ for $i$ odd, then
the quotient $f_{\bf ftm} (n)/n$ is minimized at
$n = L_{i+1}$, where it takes the value $(F_{i-1} + 1)/L_{i+1}$,
and is maximized at
$n = L_i + 1$, where it takes the value
$(F_{i-1} + 1)/(L_i + 1)$.

From the above remarks, together with 
the Binet forms for $L_i$ and $F_i$, which are
\begin{align*}
F_i &= (\alpha^i - \beta^i)/\sqrt{5} \\
L_i &= \alpha^i + \beta^i ,
\end{align*}
where $\alpha = (1+\sqrt{5})/2$ and $\beta=(1-\sqrt{5})/2$,
we obtain the following result.
\begin{corollary}
We have 
$$\liminf_{n \rightarrow \infty} f_{\bf ftm} (n)/n =
1/(\sqrt{5} \alpha^2) = 1/(\alpha+\alpha^3) \doteq 0.17082039$$
and
$$\limsup_{n \rightarrow \infty} f_{\bf ftm} (n)/n
= 1/(\sqrt{5} \alpha) = 1/(1+\alpha^2) \doteq 0.276393202.$$
\end{corollary}
These results refute Conjecture 3.1 of
\cite{Jamet&Popoli&Stoll:2021}.
The behavior of $f_{\bf ftm} (n)/n$ is depicted in
Figure~\ref{figgy} below.
\begin{figure}[htb]
\begin{center}
\includegraphics[width=5in]{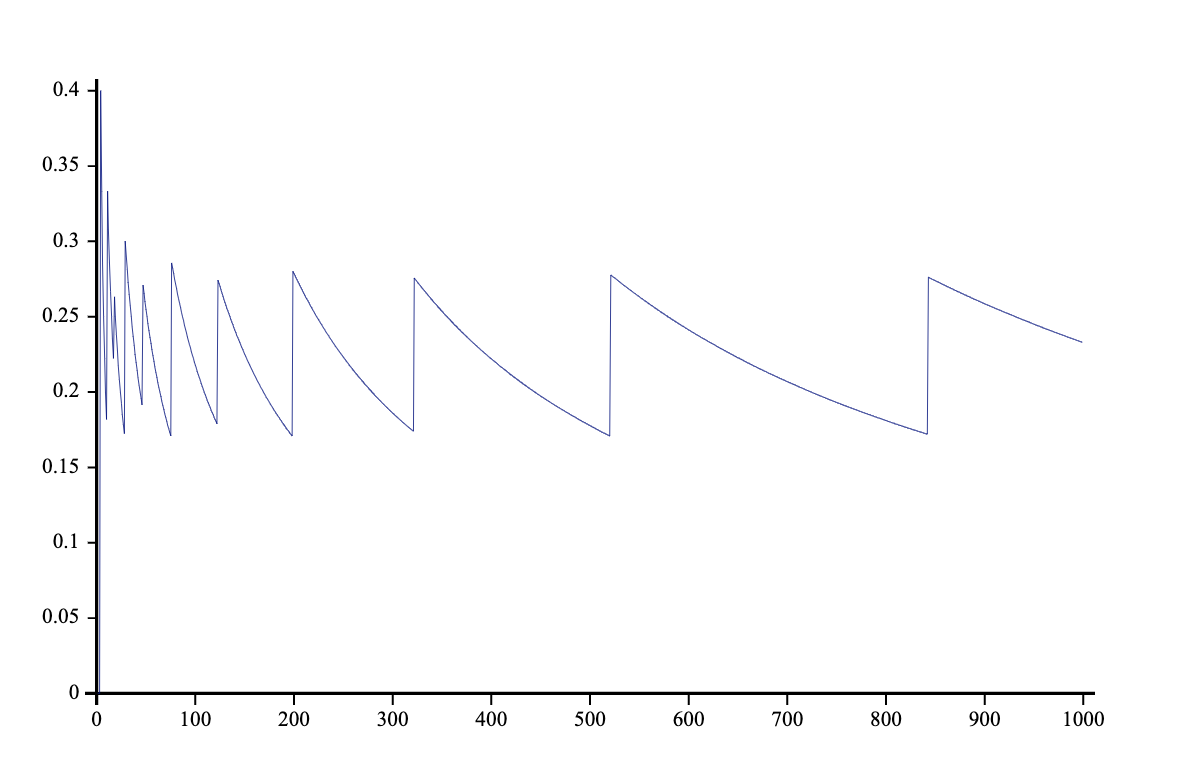}
\end{center}
\caption{The quotient $f_{\bf ftm} (n)/n$ for $2 \leq n \leq 1000$.}
\label{figgy}
\end{figure}

We can also recover Theorem 1 (that is, Theorem 1.1 of Jamet et al.~\cite{Jamet&Popoli&Stoll:2021}), with an explicit constant:
\begin{theorem}
We have $f_{\bf ftm} (n) \geq {n \over {\alpha+\alpha^3}}$ for
all $n \geq 5$.
\end{theorem}

\begin{proof}
We can verify the inequality for $n \geq 8$ using Theorem~\ref{main}.
Suppose $L_i < n \leq L_{i+1}$.

For $i\geq 4$ and even we have
$\beta^{i-2} + \beta^{i-4} \geq 0$, while $\beta^{i+1} < 0$.  
Hence
\begin{equation}
\beta^{i-2} + \beta^{i-4} \geq \sqrt{5} \cdot \beta^{i+1}.
\label{eq1}
\end{equation}
Now observe that $\sqrt{5} = \alpha + 1/\alpha$, so
adding $\alpha^{i+1} (\alpha + 1/\alpha)$ to both sides of
Eq.~\eqref{eq1}
gives
\begin{align*}
(\alpha^{i-1} - \beta^{i-1}) (\alpha+\alpha^3) &=
 \alpha^{i+2} + \alpha^i + \beta^{i-2} + \beta^{i-4}  \\
 & \geq \sqrt{5} (\alpha^{i+1} + \beta^{i+1}) .
\end{align*}
By rearranging we get
$$ {{f_{\bf ftm} (n)} \over n}  \geq  {{F_{i-1}} \over {L_{i+1}}}  = 
{{ (\alpha^{i-1} - \beta^{i-1} )/ \sqrt{5}} \over {\alpha^{i+1} + \beta^{i+1}}}
\geq {1 \over {\alpha+\alpha^3}} .$$

On the other hand, if $i\geq 5$ is odd, we have
$\sqrt{5} - \beta^{i-1} > 1$, and $\beta^{i+1} < 1$,
so
$$(\sqrt{5} - \beta^{i-1}) (\alpha + \alpha^3) \geq (\alpha + \alpha^3) \geq
\sqrt{5} \beta^{i+1}.$$
as before, adding $\alpha^{i+1} (\alpha + 1/\alpha)$  to both sides
and simplifying gives
$$ (\alpha^{i-1} - \beta^{i-1} + \sqrt{5}) (\alpha+\alpha^3) \geq 
\sqrt{5} (\alpha^{i+1} + \beta^{i+1}) ,$$
so
$$ {{f_{\bf ftm} (n)} \over n}  \geq  {{F_{i-1} + 1} \over {L_{i+1}}}
= {{ (\alpha^{i-1} - \beta^{i-1})/ \sqrt{5} + 1} \over 
{\alpha^{i+1} + \beta^{i+1}}} \geq {1 \over {\alpha+\alpha^3}} ,$$
as desired.

Finally, for $5 \leq n \leq 7$, we can easily check the inequality,
while for $n \leq 4$ it fails.  
\end{proof}

\end{document}